\newif\ifconf\conffalse
\ifconf
\documentclass{llncs}
\else
\documentclass[letterpaper,11pt]{article}
\fi

\ifconf
\else
\usepackage{amsthm}
\usepackage{fullpage}
\newtheorem{theorem}{Theorem}
\newtheorem{lemma}[theorem]{Lemma}

\newtheorem{claim}[theorem]{Claim}

\newtheorem{definition}[theorem]{Definition}

\usepackage[T1]{fontenc} 
\usepackage{lmodern}
\fi

\usepackage{etex,etoolbox}

\usepackage{amssymb}
\usepackage{amsmath}

\ifconf
\let\doendproof\endproof
\renewcommand\endproof{~\hfill\qed\doendproof}
\fi

\usepackage{color}
\usepackage{colortbl}
\usepackage{graphicx}

\newcommand{\E}{\mathop{\mathbb E}\displaylimits}
\renewcommand{\vec}[1]{\mathbf{#1}}

\makeatletter
\providecommand{\@fourthoffour}[4]{#4}
\def\fixstatement#1{%
  \AtEndEnvironment{#1}{%
    \xdef\pat@label{\expandafter\expandafter\expandafter
      \@fourthoffour\csname#1\endcsname\space\@currentlabel}}}

\globtoksblk\prooftoks{1000}
\newcounter{proofcount}

\long\def\proofatend#1\endproofatend{%
\ifconf
  \edef\next{\noexpand\begin{proof}[of Claim \pat@label]}%
\else
  \edef\next{\noexpand\section{Proof of \pat@label}}%
\fi  
  \toks\numexpr\prooftoks+\value{proofcount}\relax=\expandafter{\next#1}
  \stepcounter{proofcount}}

\def\printproofs{%
  \count@=\z@
  \loop
    \the\toks\numexpr\prooftoks+\count@\relax
    \ifnum\count@<\value{proofcount}%
    \advance\count@\@ne
  \repeat}
\makeatother

\fixstatement{theorem}
\fixstatement{lemma}
\fixstatement{fact}
\fixstatement{proposition}
\fixstatement{claim}
\fixstatement{observation}
\fixstatement{corollary}
\fixstatement{remark}
\fixstatement{conjecture}
\fixstatement{definition}

\newcommand{\eps}{\epsilon}
\newcommand{\R}{\mathbb{R}}
\newcommand{\bound}{0.5664} 

\begin{document}
\title{Online Bipartite Matching with Decomposable Weights}
\author{
Moses Charikar\thanks{Princeton University, USA, \texttt{moses@cs.princeton.edu}}
\and
Monika Henzinger\thanks{University of Vienna, Austria, \texttt{monika.henzinger@univie.ac.at}}
\and
Huy L. Nguy\~{\^{e}}n\thanks{Princeton University, USA, \texttt{hlnguyen@cs.princeton.edu}}
}

\maketitle

\begin{abstract}
We study a weighted online bipartite matching problem:
$G(V_1, V_2, E)$ is a weighted bipartite graph where $V_1$ is
known beforehand and the vertices of $V_2$ arrive online.
The goal  is to match vertices of $V_2$ as they arrive to vertices in $V_1$,
so as to maximize the sum of weights of edges in the matching.
If assignments to $V_1$ cannot be changed, no bounded competitive ratio is achievable.
We study the weighted online matching problem with {\em free disposal}, where vertices
in $V_1$ can be assigned multiple times, 
but only get credit for the maximum weight edge assigned to them over the course of
the algorithm.
For this problem, the greedy algorithm is $0.5$-competitive
and determining whether a better competitive ratio is achievable is a well known open problem.

We identify an interesting special case where the edge weights are decomposable as the
product of two factors, one corresponding to each end point of the edge.
This is analogous to the well studied related machines model in the
scheduling literature, although the objective functions are
different.
For this case of decomposable edge weights, we design a \bound\  competitive
randomized algorithm in complete bipartite graphs.
We show that such instances with decomposable weights are non-trivial by establishing
upper bounds of 0.618 for deterministic and $0.8$ for randomized
algorithms.

A tight competitive ratio of $1-1/e \approx 0.632$ was known previously for both the 0-1 case as well as 
the case where edge weights depend on the offline vertices only, but for these cases, reassignments cannot
change the quality of the solution.
Beating 0.5 for weighted matching where reassignments are necessary has been a significant challenge.
We thus give the first online algorithm with competitive ratio strictly better than 0.5 for a non-trivial
case of weighted matching with free disposal.

\end{abstract}

\section{Introduction}

In recent years, online bipartite matching problems have been intensely studied.
Matching itself is a fundamental optimization problem with several applications,
such as matching medical students to residency programs, matching men and women, matching packets to outgoing links in a router and so on.
There is a rich body of work on matching problems, yet there are basic problems we don't 
understand and we study one such question in this work.
The study of the online setting goes back to the seminal work of Karp, Vazirani and Vazirani ~\cite{KVV90}
who gave an optimal $1-1/e$ competitive algorithm for the unweighted case.
Here $G(V_1, V_2, E)$ is a bipartite graph where $V_1$ is
known beforehand and the vertices of $V_2$ arrive online.
The goal of the algorithm is to match vertices of $V_2$ as they arrive to vertices in $V_1$,
so as to maximize the size of the matching.

In the weighted case, edges have weights and the goal is to maximize the sum of weights of edges in
the matching.
In the application of assigning ad impressions to advertisers in display advertisement, the weights could represent 
the (expected) value  of an ad impression to an advertiser and the objective function for the maximum matching problem encodes the goal of assigning ad impressions to advertisers to as to maximize total value.
If assignments to $V_1$ cannot be changed and if edge weights depend on the {\em online} node to which they are adjacent, it is easy to see that no competitive ratio bounded away 
from 0 is achievable.

Feldman et al \cite{FKMMP09} introduced the {\em free disposal} setting for weighted matching, where
vertices in $V_1$ can be assigned multiple times, 
but only get credit for the maximum weight edge assigned to them over the course of
the algorithm.
(On the other hand, a vertex in $V_2$ can only be assigned at the time that it arrives with no later
reassignments permitted).
\cite{FKMMP09} argues that this is a realistic model for assigning ad impressions to advertisers.
The greedy algorithm is $0.5$ competitive for the online weighted matching problem with free disposal.
They study the weighted matching problem with capacities -- here each vertex
$v \in V_1$ is associated with a capacity $n(v)$ and gets credit for the largest $n(v)$ edge weights from
vertices in $V_2$ assigned to $v$.
They designed an algorithm with competitive ratio approaching $1-1/e$ as the capacities approach infinity.
Specifically, if all capacities are at least $k$, their algorithm gets competitive ratio $1-1/e_k$ where $e_k = (1+1/k)^k$.
If all capacities are 1, their algorithm is $1/2$-competitive.

Aggarwal et al~\cite{AGKM11} considered the online weighted bipartite matching problem where edge weights 
are only dependent on the end point in $V_1$, i.e. each vertex $v \in V_1$ has a weight $w(v)$ and the 
weight of all edges incident on $v$ is $w(v)$.
This is called the {\em vertex weighted setting}. They designed a $1-1/e$ competitive algorithm.
Their algorithm can be viewed as a generalization of the Ranking algorithm of \cite{KVV90}.

It is remarkable that some basic questions about a fundamental problem such as matching are still open in the online setting.  
Our work is motivated by the following tantalizing open problem;
{\em Is it possible to achieve a competitive ratio better than $0.5$ for weighted online matching ?}
Currently no upper bound better than $1-1/e$ is known for the setting of general weights -- in fact this bound 
holds even for the setting of 0-1 weights.
On the other hand, no algorithm with competitive ratio better than 0.5 (achieved by the greedy algorithm) is known for this problem.
By the results of \cite{FKMMP09}, the case where the capacities are all 1 seems to be the hardest case and this is what
we focus on. 

\subsection{Our results}

We identify an interesting special case of this problem where we have a complete graph between $V_1$ and $V_2$
and the edge weights are decomposable as the
product of two factors, one corresponding to each end point of the edge.
This is analogous to the well studied related machines model in the
scheduling literature \cite{AAFPW97,A98survey,BCK00,ES00}
where the load of a job of size $p$ on a machine of speed $s$ is $p/s$
although the objective functions are different.
Scheduling problems typically involving minimizing the maximum machine load (makespan) or minimizing the
$\ell_p$ norm of machine loads, where the load on a machine is the sum of loads of all jobs placed on the machine.
By contrast, in 
the problem we study,
the objective 
(phrased in machine scheduling terminology) is to maximize the sum of machine loads where the load of a
machine is the load of the largest job placed on the machine.
For this case of decomposable edge weights, we design a \bound\   competitive
algorithm (Section~\ref{sec:randalg}).
For display advertisement using a complete graph models the setting of a specific market segment (such as impressions for males between 20 and 30), where every advertiser is interested in every impression. The weight factor of the offline node $u$ can model the value that a click has for advertiser $u$, the weight factor of the online node $v$ can model the
clickthrough probability of the user to which impression $v$ is shown. Thus, the maximum weight matching in the setting we study corresponds to maximizing the sum of the expected values of all advertisers.

Our algorithm uses a now standard {\em randomized doubling technique} \cite{BN70,Gal80,CPSSSW,GK98};
however the analysis is novel and non-trivial.
We perform a recursive analysis where each step proceeds as follows:
We lower bound the profit that the algorithm derives from the fastest machine (i.e. the load of the largest job placed on it) relative to the
difference between two optimum solutions - one corresponding to the original instance and the other corresponding to a modified
instance obtained by removing this machine and all the jobs assigned to it.
This is somewhat reminiscent of, but different from the local ratio technique used to design approximation algorithms.
Finally, to exploit the randomness used by the algorithm we need to establish several structural properties of the worst case 
sequence of jobs -- this is a departure from previous applications of this {\em randomized doubling technique}.
While all previous online matching algorithms were analyzed using a {\em local}, step-by-step analysis, we use a 
{\em global} technique, i.e. we reason about the entire sequence of jobs at once. This might be useful for solving the case of online weighted matching for 
{\em general} weights.
The algorithm and analysis is presented in Section~\ref{sec:randalg} and
an outline of the analysis is presented in Section~\ref{sec:outline}.

A priori, it may seem that the setting of decomposable weights ought to be a much easier case of weighted online matching since it does not
capture the well studied setting of 0-1 weights.
We show that such instances with decomposable weights are non-trivial by establishing
an upper bound of $(\sqrt{5}-1)/2 \approx 0.618$ on the competitive ratios of deterministic algorithms (Section~\ref{sec:detlb})
and an upper bound of 0.8 on the competitive ratio of randomized algorithms (Section~\ref{sec:randlb}).
The deterministic upper bound constructs a sequence of jobs that is the solution to a certain recurrence relation.
Crucial to the success of this approach is a delicate choice of parameters to ensure that the solution of the recurrence
is oscillatory (i.e. the roots are complex).
In contrast to the setting with capacities, for which   a deterministic algorithm with competitive ratio
approaching $1-1/e \approx 0.632$ exists~\cite{FKMMP09}, our upper bound of ($\sqrt{5}-1)/2 < 1-1/e$ for deterministic algorithms shows that 
no such competitive ratio can be achieved for the decomposable case with unit capacities.
Note that the upper bound of $1-1/e$ for the unweighted case~\cite{KVV90} is for randomized algorithms and
does not apply to the setting of decomposable weights that we study here.

In contrast to the vertex weighted setting (and the special case of 0-1 weights) where reassignments
to vertices in $V_1$ cannot improve the quality of the solution, any algorithm for the decomposable 
weight setting must necessarily exploit reassignments in order to achieve a competitive ratio bounded 
away from 0.
For this class of instances, we give an upper bound approaching 0.5 for the competitive ratio of the greedy algorithm. This shows 
that for decomposable weights greedy's  performance cannot be better than  for general weights, where it is 0.5-competitive (Section~\ref{sec:greedy}).

\subsection{Related work}
Goel and Mehta~\cite{GM08} and Birnbaum and Mathieu~\cite{BM08} simplified the analysis of the Ranking algorithm considerably.
Devanur et al~\cite{DJK13} recently gave an elegant randomized primal-dual interpretation of
\cite{KVV90}; their framework also applies to the generalization to the vertex weighted setting by \cite{AGKM11}. 
Haeupler et al~\cite{HMZ11} studied online weighted matching in the stochastic setting where vertices 
from $V_2$ are drawn from a known distribution. 
The stochastic setting had been previously studied in the context of unweighted bipartite matching in a sequence
of papers \cite{FMMM09,MGS11}.
Recent work has also studied the random arrival model (for unweighted matching) where the order of arrival 
of vertices in $V_2$ is assumed to be a random permutation: In this setting, Karande, at al~\cite{KMT11} and 
Mahdian and Yan \cite{MY11}  showed that the Ranking 
algorithm of \cite{KVV90} achieves a competitive ratio better than $1-1/e$.
A couple of recent papers analyze the performance of a randomized greedy algorithm and an analog of the Ranking algorithm for matching 
in general graphs \cite{PS12,GT12}.
Another recent paper introduces a stochastic model for online matching where the goal is to maximize the number of 
successful assignments (where success is governed by a stochastic process)
\cite{MP12}.

A related model allowing cancellation of previously accepted  online nodes was studied in \cite{CFMP09,BHK09,AK09} and optimal deterministic and randomized algorithms were given. In their setting the weight of an edge depends {\em only} on the online node. Additionally in their model they decide in an online fashion only which online nodes to accept, {\em not} how to match these nodes to offline nodes. If a previously accepted node is later rejected, a non-negative cost is incurred. 
Since the actual matching is only determined after all online nodes have been seen, their model is very different from ours: 
Even if the cost of rejection of a previously accepted node is set to 0, the key difference is that they do not commit to a matching
at every step and the intended matching can change dramatically from step to step.
Thus, it does {\em not} solve the problem that we are studying.

A related problem that has been studied is online matching with preemption \cite{M05,A11,ELSW12}.
Here, the edges of a graph arrive online and the algorithm is required to maintain a subset of edges that form a matching.
Previously selected edges can be rejected (preempted) in favor of newly arrived edges.
This problem differs from the problem we study in two ways: (1) the graph is not necessarily bipartite, and
(2) edges arrive one by one. In our (classic) case, vertices arrives online and all incident edges to a newly arrived vertex $v$
are revealed when $v$ arrives.

Another generalization of online bipartite matching is the Adwords problem \cite{MSVV05,DH09}.
In addition, several online packing problems have been studied with applications to the Adwords and Display Advertisement problem
\cite{BJN07,GM08,AWY09}.
\subsection{Notation and preliminaries}

We consider the following variant of the online bipartite matching problem.
 The input is a complete bipartite graph $G = (V_1 \cup V_2, V_1 \times V_2)$ along with two weight functions $s:V_1\rightarrow \R_+$ and $w:V_2\rightarrow \R_+$. 
The weight of each edge $e=(u, v)$ is the product $s(u)\cdot w(v)$. 
At the beginning, only $s$ is given to the algorithm.
 Then, the vertices of $V_2$ arrive one by one. 
When a new vertex $v$ arrives, $w(v)$ is revealed and the algorithm has to match it to a vertex in $V_1$. At the end, the reward of each vertex $u\in V_1$ is the maximum weight assigned to $u$ times $s(u)$. The goal of the algorithm is to maximize the sum of the rewards.
To simplify the presentation we will call vertices of $V_1$ {\em machines} and vertices of $V_2$ {\em jobs}. The $s$-value of a machine $u$ will be called the
{\em speed} of the machines and the $w$-value of a job $v$ is called the {\em size} of the job. Thus, the goal of the online algorithm is to assign jobs to machines.
However, we are not studying the ``classic'' variant of the problem since we are using a different optimization function, motivated by display advertisements.

\section{Upper bound for the greedy algorithm}\label{sec:greedy}

We begin by addressing an obvious question, which is how well a greedy approach would solve our problem, and using the proof to provide some intuition for our algorithm in the next section. We analyze here the following simple  greedy algorithm:  When a job $v$ arrives,
the algorithm computes for every machine $u$ the difference between the weight of $(u,v)$ and the weight $(u,v')$, where $v'$ is the job currently assigned to $u$.
If this difference is positive for at least one machine, the job is assigned to a machine with maximum difference.

\begin{theorem}
  The competitive ratio of the greedy algorithm is at most $\frac{1}{2-\eps}$ for any $\eps > 0$.
\end{theorem}

\begin{proof}
  Consider the following instance. $V_1$ consists of a vertex $a$ with $s(a)=1$ and $t=1/\eps^2$ vertices $b_1, \ldots, b_t$ with $s(b_i) = \eps/2~\forall i$. $V_2$ consists of the following vertices arriving in the same order $d_1, \ldots, d_{1+t}$ where $w(d_i) = (1-\eps/2)^{-i}$. We will prove by induction that all vertices $d_i$ are assigned to $a$. When $d_1$ arrives, nothing is assigned so it is assigned to $a$. Assume that all the first $t$ vertices are assigned to $a$ when $d_{t+1}$ arrives. The gain by assigning $d_{i+1}$ to $a$ is $(w(d_{i+1})-w(d_i))s(a) = \eps(1-\eps/2)^{-i-1}/2$. The gain by assigning $d_{i+1}$ to some $b_j$ is $w(d_{i+1})s(b_j) = \eps(1-\eps/2)^{-i-1}/2$. Thus, the algorithm can assign $d_{i+1}$ to $a$. The total reward of the algorithm is $(1-\eps/2)^{-1-t}$. The optimal solution is to assign $d_{1+t}$ to $a$ and the rest to $b_i$'s, getting $(1-\eps/2)^{-1-t}+(1-\eps/2)^{-t} - 1 \ge (2-\eps)(1-\eps/2)^{-1-t}$. Thus, the competitive ratio is at most $\frac{1}{2-\eps}$.
\end{proof}

The instance used in the proof above suggests some of the complications an algorithm has to deal with in the setting of decomposable weights: in order to have competitive ratio bounded away from $0.5$, an online algorithm must necessarily place some jobs on the slow machines. 
In fact it is possible to design an algorithm with competitive ratio bounded away from $0.5$ for the specific set of machines used in this proof 
(for any sequence of jobs). 
The idea is to ensure that a job is placed on the fast machine only if its size is larger than $(1+\gamma)$ times the size of the largest job currently on the fast machine (for an appropriately chosen parameter $\gamma$).
Such a strategy works for any set of machines consisting of one fast machine and several slow machines of the same speed.
However, we do not know how to generalize this approach to an arbitrary set of machines.
Still, this strategy (i.e. ensuring that jobs placed on a machine increase in size geometrically)
was one of the motivations behind the design of the randomized online algorithm to be presented next.
\section{Randomized algorithm} \label{sec:randalg}
We now describe our randomized algorithm which uses a parameter $c$ we will specify later:
The algorithm picks values $x_i \in (0,1]$ uniformly and at random,
independently for each machine $i$. 
Each job of weight $w$ considered by machine $i$ is placed in the unique interval $w\in (c^{k+x_i}, c^{k+1+x_i}]$ 
where $k$ ranges over all integers. 
When a new job $w$ arrives, the algorithm checks the machines in the order of decreasing speed (with ties broken in an arbitrary but fixed way). For 
machine $i$ it first determines the unique interval into which $w$ falls, which depends on its choice of $x_i$.
If the machine currently does not have a job in this or a bigger interval (with larger $k$), $w$ is assigned to $i$ and the algorithm stops, otherwise the algorithm checks the next machine.

The following function arises in our analysis:
\begin{definition}
Define $\displaystyle h(c) = 1-\frac{1}{\beta}W\left( \frac{\beta e^\beta}{c} \right)$\\
where $\beta =  \frac{c \ln(c)}{c-1}-1$ and $W()$ is the Lambert W function (i.e. inverse of $f(x)=x e^x$).
\end{definition}

\newcommand{\expression}{\min \left( \frac{c-1}{c \ln(c)}, h(c) \right)}

We will prove the following theorem:
\begin{theorem}
For $c \ge e$, the randomized algorithm has competitive ratio
$\expression$.
In particular, for $c=3.55829$, 
the randomized algorithm has a competitive ratio $\bound$.
\end{theorem}

\subsection{Analysis Outline}\label{sec:outline}
We briefly outline the analysis strategy before describing the details.
An instance of the problem consists of a set of jobs and a set of machines.
The (offline) optimal solution to an instance is obtained by ordering machines from
fastest to slowest, ordering jobs from largest to smallest and assigning the
$i$th largest job to the $i$th fastest machine.
Say the machines are numbered $1,2, \ldots n$, from fastest to slowest.
Let $OPT_i$ denote the value of the optimal solution for the instance seen by the 
machines from $i$ onwards, i.e. the instance consisting of machines $i, i+1, \ldots n$, 
and the set of jobs passed by the $(i-1)$st machine to the $i$th machine in the online algorithm.
Then $OPT_1 = OPT$, the value of the optimal solution for the original instance.
Even though we defined $OPT_i$ to be the value of the optimal solution, we will sometimes
use $OPT_i$ to denote the optimal assignment, although the meaning will be clear from context.
Define $OPT_{n+1}$ to be 0.
For $2 \leq i \leq n$, $OPT_i$ is a random variable that depends on the random values
$x_{i'}$ picked by the algorithm for $i' < i$.
In the analysis, we will define random variables $\Delta_i$ such that $\Delta_i \geq OPT_i - OPT_{i+1}$
(see Lemma~\ref{lem:delta_i} later).
Let $A_i$ denote the profit of the online algorithm derived from machine $i$
(i.e. the size of the largest job assigned to machine $i$ times the speed of  machine $i$).
Let $A = \sum_{i=1}^n A_i$ be the value of the solution produced by the online algorithm.
We will prove that for $1 \leq i \leq n$,
\begin{align}
\E[A_i] \geq \alpha \E[\Delta_i] \geq \alpha (\E[OPT_i] - \E[OPT_{i+1}])  \label{eq:ADelta}
\end{align}
for a suitable choice of $\alpha > 0.5$.
The expectations in (\ref{eq:ADelta}) are taken over the random choices of machine $1, \ldots i$.
Note that $OPT_i - OPT_{i+1}$ is a random variable, but the sum of these quantities for $1 \leq i \leq n$ is 
$OPT_1 - OPT_{n+1} = OPT$, a deterministic quantity.
Summing up (\ref{eq:ADelta}) over $i =1, \ldots n$, we get $\E[A] \geq \alpha \cdot OPT$, 
proving that the algorithm gives an $\alpha$ approximation.

Inequality (\ref{eq:ADelta}) applies to a recursive application of the algorithm to the subinstance
consisting of machines $i, \ldots n$ and the jobs passed from machine $i-1$ to machine $i$.
The subinstance is a function of the random choices made by the first $i-1$ machines.
We will prove that for any instance of the random choices made by the first $i-1$ machines,
\begin{align}
\E[A_i] \geq \alpha \E[\Delta_i]. \label{eq:ADelta2}
\end{align}
Here, the expectation is taken over the random choice of machine $i$.
(\ref{eq:ADelta2}) immediately implies (\ref{eq:ADelta}) by taking expectation over the
random choices made by the first $i-1$ machines.

We need to establish (\ref{eq:ADelta2}).
In fact, it suffices to do this for $i=1$ and the proof applies to all values of $i$ since (\ref{eq:ADelta2})
is a statement about a recursive application of the algorithm.
Wlog, we normalize so that the fastest machine has speed $1$
and the largest job is $c$. 
Note that this is done by simply multiplying all machine speeds by a suitable factor and all job sizes by a suitable factor
-- both the LHS and the RHS of (\ref{eq:ADelta2}) are scaled by the same quantity.

In order to compare $\Delta_1$ with the profit of the algorithm, we
decompose the instance 
into a convex combination of simpler threshold 
instances in Lemma~4. Here, the speeds are either all the same or
take only two different values, 0 and 1.
It suffices to compare the profit of the algorithm to OPT on 
such threshold instances.

Intuitively, if there are so few fast machines that even a relatively
large job (job of weight at least 1) got assigned to a slow machine in
OPT, then the original instance is mostly comparable to the threshold instance
where only a few machines have speed 1 and the rest have speed 0. Even
if the fastest machine gets jobs assigned to machines of speed 0 in OPT,
this does not affect the profit of the algorithm relative to OPT because OPT
does not profit from these jobs either.
Thus we only care about jobs of weight at least 1.
Because a single machine can get at most two jobs of value in the range $[1,c]$,
handling this case only requires analyzing at most two jobs.
The proof for this case is contained in Lemma~\ref{lem:c^t}.

On the other hand, if there are a lot of fast machines so that all large
jobs are assigned to fast machines in OPT, then the original instance is
comparable to the threshold instance where all machines have speed 1. In this
case, the fastest machine can get assigned many jobs that all contribute
to OPT. However, because all speeds are the same, we can
deduce the worst possible sequence of jobs: after the first few jobs,
all other jobs have weights forming a geometric sequence. The rest of
the proof is to analyze the algorithm on this specific sequence. The detailed
proof is contained in Lemma~\ref{lem:9}.

The proofs of both Lemmata~\ref{lem:c^t} and~\ref{lem:9} use the
decomposable structure of the edge weights.

\subsection{Analysis Details}

Recall that $OPT_1$ is the value of the optimal solution for the instance, and $OPT_2$ is the value of the optimal
solution for the subinstance seen by machine 2 onwards.
Assume wlog that all job sizes are distinct (by perturbing job sizes infinitesimally).
For $y \leq c$, let $j(y)$ be the size of the largest job $\leq y$ or 0, if no such  job  exists.
Let $s(y)$ be the speed of the machine in the optimal solution that $j(y)$ is assigned to
or 0 if $j(y) = 0$.
If there is a job of size $y$ then $s(y)$ is the speed of the machine in the optimal solution
that this job is assigned to.
Note that $s(y) \in [0,1]$ is monotone increasing with $s(c)=1$. 
We refer to the function $s$ as the {\em speed profile}.
Note that $s$ is not a random variable.
Let the {\em assignment sequence} $\vec{w} = (w,w_1,w_2,\ldots)$ denote the set of jobs assigned to the fastest 
machine by the algorithm where $w > w_1 > w_2 > \ldots$.
Let $\max(\vec{w})$ denote the maximum element in the sequence $\vec{w}$, i.e. $\max(\vec{w}) = w$.
In Lemma~3, we bound $OPT_1 - OPT_2$ by a function that depends {\em
only} on $\vec{w}$, $s$, and $c$. Such a bound is possible because of
the fact that any job can be assigned to any machine, i.e. the graph is a complete graph. 
The value we take for the aforementioned random variable $\Delta_1$ turns out to be
exactly this bound.

\begin{lemma}
$OPT_1 - OPT_2 \leq c - (c-w) s(w) + \sum_{k \geq 1} w_k \cdot s(w_k)$
\label{lem:delta_i}
\end{lemma}
\begin{proof}
Let $I_1$, $I_2$ be the instances corresponding to $OPT_1$ and $OPT_2$.
$I_2$ is obtained from $I_1$ by removing the fastest machine and the set of jobs
that are assigned to the fastest machine by the algorithm.
Let us consider changing $I_1$ to $I_2$ in two steps:
(1) Remove the fastest machine and the largest job $w$ assigned by the algorithm to the fastest machine.
(2) Remove the jobs $w_1, w_2, \ldots$.
For each step, we will bound the change in the value of the optimal solution resulting in a feasible solution
for $I_2$
and computing its value -- this will be a lower bound for $OPT_2$.

First we analyze Step 1:
$OPT_1$ assigns the largest job $c$ to the fastest machine, contributing $c$ to its value.
The algorithm assigns $w$ to the fastest machine instead of $c$.
In $OPT_1$, $w$ was assigned to a machine of speed $s(w)$.
When we remove $w$ and the fastest machine from $I_1$, one possible assignment to the 
resulting instance is obtained by placing $c$ on the machine of speed $s(w)$.
The value of the resulting solution is lower by exactly $(c+w \cdot s(w)) - c \cdot s(w) = c-(c-w)s(w)$.
  
Next, we analyze Step 2:
Jobs $w_1, w_2, \ldots$ were assigned to machines of speeds $s(w_1), s(w_2), \ldots$ in $OPT_1$.
When we remove jobs $w_1, w_2, \ldots$, one feasible assignment for the resulting instance is
simply not to assign any jobs to the machines $s(w_1), s(w_2), \ldots$, and keep all other
assignments unchanged.
The value of the solution drops by exactly $\sum_{k \geq 1} w_k \cdot s(w_k)$.

Thus we exhibited a feasible solution to instance $I_2$ of value $V$ where
$$OPT_1 - V = c - (c-w) s(w) + \sum_{k \geq 1} w_k \cdot s(w_k).$$
But $OPT_2 \geq V$.
Hence, the lemma follows.
\end{proof}
We define the random variable $\Delta_1$, a function of the assignment sequence $\vec{w}$ and the speed profile $s$, to be
$$\Delta_1(\vec{w},s) = c - (c-w) s(w) + \sum_{k \geq 1} w_k \cdot s(w_k).$$
As defined, $\Delta_1(\vec{w},s) \geq OPT_1 - OPT_2$.
We note that even though $OPT_1$ and $OPT_2$ are functions of all the jobs in the instance,
$\Delta_1$ only depends on the subset of jobs assigned to the fastest machine by the algorithm.
Our goal is to show $\E[A_1] = E[\max(\vec{w})] \geq \alpha \E[\Delta_1]$.

First, we argue that it suffices to restrict our analysis to a simple set of step function speed profiles
$s_t$, $0 \leq t \leq 1$:
For $t \in (0,1]$, $s_t(y)=1$ for $y \in [c^t,c]$ and $s_t(y)=0$ for $y< c^t$.
For $t=0$, $s_0(y) = 1$ for all $y \leq c$.

\begin{lemma}\label{lem:decompose-step-functions}
Suppose that for $t=0$ and for all $t \in (0,1]$ such that there exists a job of weight $c^t$, we have
\begin{align}
\E[\max(\vec{w})] &\geq {\alpha} \E[\Delta_1(\vec{w},s_t)]  \label{eq:ADeltaStep}
\end{align}
Then, $\E[\max(\vec{w})] \geq {\alpha} (\E[OPT_1] - \E[OPT_2])$.
\end{lemma}
\begin{proof}
Consider function $s'(y)$ defined as follows:
$s'(y) = s(y)$ for $y \in [1,c]$ and $s'(y) = s(1)$ for $y < 1$.
Note that $s'$ is not a random variable.
We claim that $\Delta_1(\vec{w},s') \geq \Delta_1(\vec{w},s)$.
Since the largest job assigned to the fastest machine is $w \in [1,c]$, the $s(w)$ term is unchanged in 
going from $\Delta_1(\vec{w},s)$ to $\Delta_1(\vec{w},s')$.
Further, the $s(w_k)$ terms in $(\vec{w},s')$ are $\geq$ the corresponding terms in $\Delta_1(\vec{w},s)$.

It is easy to see that $s'$ is a convex combination of the step functions $s_t$, $0 \leq t \leq 1$. 
More specifically, $s' = \sum_t p_t s_t$ for suitably chosen coefficients $p_t$ such that (a) $\sum_t p_t=1$
and (b) $p_t = 0$ if $t>0$ and no  job with weight $c^t$ exists.

For a fixed assignment sequence $\vec{w}$,
note that $\Delta_1(\vec{w},s') = \Delta_1(\vec{w},\sum_t p_t s_t) = \sum_t p_t \cdot \Delta_1(\vec{w},s_t)$.
Hence, for a distribution over assignment sequences $\vec{w}$,
$$\E[\Delta_1(\vec{w},s')] = \E[OPT_1(\vec{w},\sum_t p_t s_t)] = \sum_t p_t \cdot \E[\Delta_1(\vec{w},s_t)]$$

Now, suppose that for all $0 < t \leq 1$ such that there exists a job of weight $c^t$ 
and for $t=0$
\begin{align*}
\E[\max(\vec{w})] &\geq {\alpha} \E[\Delta_1(\vec{w},s_t)].\\
\intertext{This implies that}
\E[\max(\vec{w})] &\geq {\alpha} \sum_t p_t \cdot \E[\Delta_1(\vec{w},s_t)] = {\alpha} \E[\Delta_1(\vec{w},s')]\\
& \geq {\alpha}\E[\Delta_1(\vec{w},s)]
\geq \alpha (\E[OPT_1] - \E[OPT_2])
\end{align*}
\end{proof}

Note that since we scaled job sizes, the thresholds (i.e interval boundaries) $c^{k+x_1}$ should also be scaled by the same 
quantity (say $\gamma$).
After scaling, let $x \in (0,1]$ be such that $c^x$ is the unique threshold from the set $\{\gamma c^{k+x_1}, k \text{\  integer}\}$ 
in $(1,c]$.
Since $x_1$ is uniformly distributed in $(0,1]$, $x$ is also uniformly distributed in $(0,1]$.
Having defined $x$ thus, the interval boundaries picked by the algorithm for the fastest machine are $c^{x+k}$ for integers $k$.

We prove (\ref{eq:ADeltaStep}) for $\alpha = \expression$ in two separate lemmata, 
one for the case $t>0$ (Lemma~\ref{lem:c^t}) and the other for the case $t=0$ (Lemma~\ref{lem:9}).
Recall that
the expression for $\Delta_1$ only depends on the subset of jobs assigned to the fastest machine.
We call a job a {\em local maximum} if it is larger than all jobs preceding it.
Since the algorithm assigns a new job to the fastest machine  if and only if it falls in a larger interval than the current
largest job, it follows that any job assigned to the fastest machine must be a local maximum.

Define $m_S(y)$ to be the minimum job in the sequence of all local maxima in the range $(y, cy]$, i.e., the first job larger than $y$ and at most $cy$, if such a job exists and $0$ otherwise. We use $m_S(y)$ in two ways.
(1) We define $u_0 = m_S(1)$. Note that $u_0$ is not a random variable. We use $u_0$ in Lemma~\ref{lem:c^t} to prove the desired statement for $t>0$. 
Specifically,  
we use $u_0$ to compute (i) a lower bound for $\E[w]$ as a function of $u_0$ (and not of any other jobs)
and (ii) an upper bound for $\E[\Delta_1]$ as a function of $u_0$. Combining (i) and (ii) we prove that the desired inequality holds for all $u_0$.
(2) In Lemma~\ref{lem:9} we bound
$\E[\sum_{k\ge 1}w_k]$ by a sum of $m_S(y)$ over suitable values of $y$. This simplifies the analysis since the elements in the
subsequence of {\em all} local maxima are {\em not} random variables, while the values in $\vec{w}$ are
random variables.

We first prove some simple properties of $u_0$ that we will use:
\begin{claim}
(1) $u_0 \leq w$ and  (2) $u_0 \geq w_1$.
\end{claim}
\begin{proof}
$u_0 \leq w$ as $u_0$ is the minimum element in the sequence of all local maxima in $(1,c]$
and $w$ is the element from the interval $(1,c]$ picked by the algorithm. 

$w_1$ is the minimum element in the sequence of local maxima in the range $(c^{x-k-1}, c^{x-k}]$ for $x \in (0,1]$ and
$k$ a non-negative integer. Either $u_0 \ge c^{x-k} \ge w_1$, or $u_0$ also falls into $(c^{x-k-1}, c^{x-k}]$ and
 $u_0 \geq w_1$ follows from the fact that $w_1$ is the smallest local maximum in this range, while $u_0$ is an arbitrary local  maximum in this range.
\end{proof}

The next lemmata conclude our algorithm analysis.
\begin{lemma}\label{lem:c^t} 
For $c \ge e$, $t\in (0,1]$
such that there exists a job of weight $c^t$,
and $\alpha = \expression$,
we have
$$\alpha\E[\Delta_1(\vec{w},s_{t})] \le \E[\max(\vec{w})]$$
\end{lemma}
\begin{proof}
Because there is a job with weight $c^t$, it must be the case that $u_0 \le c^t$. 
As $w$ is the job placed by the algorithm on the fastest machine, $w$ is in the same interval as $c$ for
any choice of the random value $x \in (0,1]$. Thus, $w \ge c^x$. 
As $c \ge w > c w_2$ it follows that
$w_k < 1 \leq c^t$ for all $k >1$ and, thus, $s_t(w_k) = 0$ for all $k > 1$. Hence $\sum_{k > 1} w_k \cdot s_t(w_k) = 0$.
To analyze $\E[\Delta_1]$ we have to consider two cases, depending on whether $u_0 = c^t$ (and hence $w_1$ might contribute to $\E[\Delta_1]$) or
whether $u_0 < c^t$ (and, thus, $s(w_1) = 0$ and $w_1$ does not contribute to $\E[\Delta_1]$).

{\bf Case 1:} $u_0 = c^t$.
Since $w \ge u_0$ it holds that $s_{t}(w)=1$ for all choices of $x$. Thus we have
$$\E[c-(c-w)s_{t}(w)] = \E[w]$$

As discussed above, $\sum_{k > 1} w_k \cdot s_t(w_k) = 0$ and, thus, the only contribution to $\E[\sum_{k\ge 1}w_k\cdot s_{t}(w_k)]$ is from $w_1$. Additionally $s_{t}(w_1) = 1$ only if   $w_1=u_0=c^t$, and this only happens when $x$ is chosen such that $x\ge t$. Thus,
$$\E[\sum_{k\ge 1}w_k\cdot s_{t}(w_k)] \le (1-t)c^t$$
Note that $w \ge \max(c^x, c^t)$. Thus we have
\begin{align*}
\E[w] &\ge \int_t^1 c^x dx + \int_0^{t} c^t dx
\ge \frac{c-c^t}{\ln c} + tc^t
\end{align*}

In this case, we want to show 
$$\alpha\E[\Delta_1] \le \E[w].$$
This holds if
$$\alpha\left(c^t + \frac{c-c^t}{\ln c}\right) \le \frac{c-c^t}{\ln c}+tc^t$$

Since $u_0 = c^t$, this inequality follows for all $\alpha \le h(c)$ from Inequality~\ref{eq:ratio-all1} below.

{\bf Case 2:} $u_0 < c^t$.
As $w_k \le u_0$ for $k \ge 1$, in this case, for all choices of $x$, all speeds $s_{t}(w_1)=s_{t}(w_2)=\ldots = 0$ so $\E[\sum_{k\ge 1}w_k\cdot s_{t}(w_k)]=0$. 
Thus, it suffices to show that $\E[\alpha (c - (c-w)s_t(w))] \le \E[w]$, or equivalently that
$\alpha \E[c(1-s_t(w))] \le \E[w - \alpha w s_t(w)]$.
 
 Let $c^z$ be the greatest local maximum that is smaller than $c^t$.
If $x > z$, then $w \ge c^t$ and, thus, $s_{t}(w) = 1$. 
If $x \le z$, then $w$ is the first local maximum greater than $c^x$, while
$c^z$ is a local maximum greater than $c^x$. Thus, it is either equal to $w$ or a later local maximum, which by the definition of local maximum implies that it is larger than $w$. Hence,
$w\le c^z < c^t$ and thus, $s_{t}(w) = 0$. Therefore,
$$\E[c(1-s_{t}(w))] \le c \int_0^z 1 dx = cz$$
We also have
\begin{align*}
\E[(1-\alpha s_{t}(w))w] &\ge (1-\alpha)\int_t^1 c^x dx +(1-\alpha)\int_{z}^t c^t dx+ \int_{\log_c u_0}^{z} c^x dx + \int_0^{\log_c u_0} u_0 dx\\
&=(1-\alpha)\frac{c-c^t}{\ln c}+(1-\alpha)\frac{c^t(t-z)}{\ln c}+\frac{c^z-u_0}{\ln c}+u_0\log_c u_0\\
&=: V(t, u_0)
\end{align*}

Thus it suffices to show that $\alpha \E[c(1-s_{t}(w))] \le V(t, u_0)~\forall z\in[0,1],t\in[z,1], u_0 \in[1,c^z]$. For any fixed $t$ and $z$, the value of $u_0$ minimizing $V(t, u_0)$ is $u_0=1$. After fixing $u_0=1$, we have 
$$V(t, 1) = (1-\alpha)\frac{c-c^t}{\ln c}+(1-\alpha)\frac{c^t(t-z)}{\ln c}+\frac{c^z-1}{\ln c}$$
Therefore,
$$\frac{\partial V(t,1)}{\partial t} = (1-\alpha)\frac{c^t + tc^t/\ln c - c^t/\ln c - zc^t/\ln c}{\ln c}$$
Notice that $\partial V(t,1)/\partial t$ is non-negative for all $t\in [z,1]$ if $c\ge e$.
Therefore, for any $c\ge e$, it suffices to consider only $u_0=1, t=z$ and prove that
$$\alpha \left(ct + \frac{c-c^t}{\ln c}\right) \le \frac{c-1}{\ln c}$$
for $\alpha$ as large as possible. The following claim shows that this inequality holds for 
$\alpha \le  \frac{c-1}{c\ln c}$.

\begin{claim}
For $c\ge e$,
$$\frac{\frac{c-1}{\ln c}}{ct + \frac{c-c^t}{\ln c}} \ge \frac{c-1}{c\ln c}~\forall t\in[0,1]$$
\end{claim}
\begin{proof}
  Consider $f(t) = ct + \frac{c-c^t}{\ln c}$. We have $f'(t) = c - \frac{c^t}{\ln^2 c}\ge 0~\forall t\in [0,1]$. Thus, the maximum $f(t)$ is achieved when $t=1$ and $f(1) = c$. Therefore,
$$\frac{\frac{c-1}{\ln c}}{ct + \frac{c-c^t}{\ln c}} \ge \frac{c-1}{c\ln c}~\forall t\in[0,1]$$
\end{proof}
Thus, altogether the lemma holds for $\alpha = \expression$.
\end{proof}

\begin{lemma}\label{lem:9}
For $s_0(x)\equiv 1$ 
and $\alpha = h(c)$,
we have $\E[\max(\vec{w})] \ge \alpha \E[\Delta_1(\vec{w},s_{0})]$.
\end{lemma}

\begin{proof}
Since $s_0(w)=1$ for all choices of $x$, it holds that
$\E[c-(c-w)s_{t}(w)] = \E[w]$.
Thus we need to show $\alpha\E[\sum_{k\ge 1} w_k] \le (1-\alpha)\E[w]$. 
As $w \ge \max(u_0, c^x)$ we have the following lower bound for $\E[w]$:
$$\E[w] \ge \int_{\log_c u_0}^1 c^x dx + \int_{0}^{\log_c u_0} u_0 dx =  \frac{c-u_0}{\ln c} + u_0 \log_c u_0$$

Now, to prove the inequality, we only need to bound from above $\E[\sum_{k\ge 1} w_k]$ for a fixed $u_0\in [1,c]$. We can write $\E[\sum_{k\ge 1}w_k]$ in terms of $m_S(x)$ as follows.
$$\E[\sum_{k\ge 1}w_k] \le \sum_{i=1}^{\infty} \int_{0}^1 m_S(c^{-i+x}) dx = \int_{-\infty}^0 m_S(c^x)dx =: B_S$$

The following claims analyze the structure of the jobs smaller than $u_0$ in the worst case, i.e., if $S$ maximizes $B_S$.

\begin{claim}
For any sequence $S$ of all local maxima where there are 2 consecutive local maxima $u_0\ge w'_u\ge w'_{u+1}$ with $w'_u > c w'_{u+1}$ there is a sequence $S'$ with $B_S$ at least as large and no such pair of consecutive local maxima.
\end{claim}
\begin{proof}
Add a new local maximum of weight $w'_u/c$ to $S$ to form $S'$. Notice that $m_S(x) \le m_{S'}(x)~\forall x$. This argument can be repeated until there is no pair of consecutive local maxima with ratio greater than $c$.
\end{proof}

\begin{claim}
Consider a sequence of all local maxima $S$ with 3 consecutive local maxima $u_0\ge w'_u \ge w'_{u+1}\ge w'_{u+2}$ where $w'_u \le cw'_{u+2}$. After removing $w'_{u+1}$, the resulting sequence $S'$ has $B_{S'} \ge B_{S}$.
\end{claim}
\begin{proof}
For all $y\not\in [w'_{u+2}, w'_u]$, we have $m_{S'}(y) = m_{S}(y)$.
For all $y\in [w'_{u+2}, w'_u]$, we have $m_{S'}(y) \ge m_{S}(y)$. Thus, $B_{S'}\ge B_{S}$.
\end{proof}

\begin{claim}
 Consider a sequence of all local maxima $S$ containing $w'_u > w'_{u+1} > \cdots> w'_{v}$ satisfying

(1) $u_0\ge w'_u$

(2) $w'_{u+i} = c^{1-i}w'_{u+1}~\forall 1\le i < v-u$

(3) $w'_u/w'_{u+1} \le c$, and

(4) $w'_{v-1}/w'_v \le c$

\noindent
Then either one of the following conditions applies

(1)  $w'_u/w'_{u+1} = c$, or

(2)  $w'_{v-1}/w'_{v} = c$, or

there is a sequence $S'$ with at most the same number of local maxima and $B_{S'} > B_{S}$.
\end{claim}

\begin{proof}
Assume that none of the conditions applies. We will show it is possible to move the jobs to form a sequence $S'$ with $B_{S'}\ge B_{S}$.

We consider the effect of moving $z=w'_{u+1}$ while maintaining the relation $w'_{u+j}=c^{1-j}w'_{u+1}~\forall 1\le j<v-u$. We have
\begin{align*}
B_S &= \int_{-\infty}^{\log_c w'_v} m_S(x)dx + \int_{\log_c w'_u}^1 m_S(x)dx\\
&\qquad +\int_{\log_c w'_v}^{\log_c w'_u} m_S(x)dx\\
 &= T + \int_{\log_c z}^{\log_c w'_u} w'_u dx + \sum_{j=0}^{v-u-3}\int_{\log_c c^{-j-1}z}^{\log_c c^{-j}z}c^{-j}zdx\\
 &\qquad  +\int_{\log_c w'_{v}}^{\log_c c^{u+2-v}z} c^{u+2-v}z dx
\end{align*}
where $T$ is a function that does not depend on $z$. Furthermore, we have
\begin{align*}
\frac{\partial B_S}{\partial z} &= -w'_u \frac{1}{z\ln c} + \frac{(1-c^{u+2-v})}{1-c^{-1}}\\
&\qquad +c^{u+2-v}(u+2-v+\frac{1 + \ln z}{\ln c} - \frac{\ln w'_{v}}{\ln c})
\end{align*}

Notice that $\partial B_S/\partial z$ is monotonically increasing so the maximum of $B_S$ is achieved at an extreme point, which is either $w'_u, c^{v-u-2}w'_{v}, w'_u/c, c^{v-u-1}w'_{v}$. When $z=w'_u$ or $z=c^{v-u-2}w'_{v}$, there are 2 jobs of the same weight and we can remove one without changing $B_S$. If
$z = w'_u/c$, the first condition in the lemma holds.
If $z = c^{v-u-1}w'_v$, then the second condition ($w_{v-1}/w_{v} = c$) holds as $z = c^{v-u} w_{v-1} $ by the assumptions of the lemma. Thus, the conclusion follows from an inductive argument on the number of jobs.
\end{proof}

By the above claims, the only sequences we need to consider to prove
Lemma~\ref{lem:9} are of the form 
$$u_0, u_0/c, \ldots, u_0/c^m, v/c^{m+1}, v/c^{m+2}, \ldots$$
where $u_0 \le v$, i.e., all pairs of consecutive jobs have ratio exactly $c$ except for possibly one pair. Thus it holds that
\begin{align*}
B_S &= \int_{\log_c u_0}^1 u_0 dx +\sum_{i=1}^{m}\int_{0}^1 u_0 c^{-i}dx \\
&\qquad +\int_{\log_c v}^{\log_c u_0} u_0 c^{-m}dx+ \sum_{i=m+1}^{\infty} \int_0^1 v c^{-i} dx\\
&=(1-\log_c u_0)u_0 + \frac{(1-c^{-m})u_0}{c-1}\\
&\qquad +u_0 c^{-m}\log_c (u_0/v)+\frac{v c^{-m}}{c-1}
\end{align*}
Notice that $\frac{\partial B_S}{\partial v} = -\frac{u_0 c^{-m}}{v \ln c} + \frac{c^{-m}}{c-1}$ is monotonically increasing so the choice of $v$ maximizing $B_S$ is either $v=u_0$ or $v=c$. The following lemma proves that the value of $B_S$ when $v = u_0$ is larger
than the value of $B_S$ when $v = c$. Thus $B_S$ is maximized when $v = u_0$, i.e., {\em all} pairs of
consecutive jobs less than $u_0$ have ratio exactly $c$.
\begin{claim}
$$\frac{u_0}{c-1} \ge u_0 \log_c (u_0/c) + \frac{c}{c-1}~\forall u_0\in[1,c]$$
\end{claim}
\begin{proof}
  Let $f(x) = \frac{x}{c-1} - x\log_c (x/c) + \frac{c}{c-1}$. We have $f(1)\ge0$ and $f(c) \ge 0$. Also notice that
$$f'(x) = \frac{1}{c-1} - \frac{\ln x + 1}{\ln c} $$
is monotonically decreasing in $x$ so the minimum of $f(x)$ is achieved at the extreme points. In other words, $f(x) \ge 0~\forall x\in[1,c]$.
\end{proof}

It follows that $B_S = (1-\log_c u_0)u_0 +\frac{u_0}{c-1}$. Thus, we need to show
for $\alpha = h(c)$
\begin{equation}
\alpha\left( (1-\log_c u_0)u_0 +\frac{u_0}{c-1} \right) \le (1-\alpha) \left( \frac{c-u_0}{\ln c} + u_0 \log_c u_0 \right)
\end{equation}
or equivalently
\begin{equation}
\alpha\left(\frac{u_0 c}{c-1} +\frac{c-u_0}{\ln c}\right) \le \frac{c-u_0+u_0 \ln u_0}{\ln c} \label{eq:ratio-all1}
\end{equation}
Let $\beta = \frac{c\ln c}{c-1}-1$. We can rewrite the above inequality as
$$\alpha(\beta u_0 +c) \le c-u_0+u_0 \ln u_0$$

\begin{claim}
For $\alpha = h(c)$,
$$\alpha(\beta u_0 +c) \le c-u_0+u_0 \ln u_0$$
\end{claim}
\begin{proof}
Let $f(u_0) = c-u_0+u_0\ln u_0$ and $g(u_0) = \beta u_0+c$. 
We need to show that $f(u_0)/g(u_0) \ge \alpha$.
Note that $(f(u_0)/g(u_0))' = 0$ if and only if
\begin{align*}
0 &= f'(u_0)g(u_0) - g'(u_0)f(u_0)\\
 &= \ln u_0 (\beta u_0+c) - \beta(c-u_0+u_0\ln u_0)\\
 &= c\ln u_0 -\beta c + \beta u_0 
\end{align*}
The above expression is monotonically increasing in $u_0$. We show next that it
has a root. Thus, $f(u_0)/g(u_0)$ is minimized at this root. Let $y = u_0 \beta / c$. The above equation is equivalent to
$$\beta +\ln(\beta/c)= \ln y + y$$
The solution of this equation is $y = W\left(\frac{\beta e^\beta}{c}\right)$, where $W$ is the Lambert W function. Thus, $u_0 = \frac{c}{\beta}W\left(\frac{\beta e^\beta}{c}\right)$. Substituting the identity $\ln u_0 = \beta - \beta u_0/c$ into $f(u_0)/g(u_0)$, we get
\begin{align*}
\frac{f(u_0)}{g(u_0)} &= \frac{c - u_0 + u_0(\beta-\beta u_0/c)}{\beta u_0+c} = 1-\frac{u_0}{c} \\
&= 1-\frac{1}{\beta}W\left(\frac{\beta e^\beta}{c}\right) = h(c)
\end{align*}

\ 
\end{proof}
Thus, inequality (\ref{eq:ratio-all1}) holds for $\alpha = h(c)$. This completes the proof of Lemma~\ref{lem:9}.
\end{proof}
\section{Upper bound for deterministic algorithms} \label{sec:detlb}
To prove an upper bound of $a$, we construct an instance such that any deterministic algorithm has competitive ratio at most $a$
for some prefix of the request sequence.
The instance has one {\em fast} machine of speed $r>1$ and $n$ {\em slow} machines of speed $1$.
The request sequence has non-decreasing job sizes satisfying a certain oscillatory recurrence relation.

\begin{theorem}\label{thm:detlb}
The competitive ratio of any deterministic algorithm is at most $(\sqrt{5}-1)/2+\eps \approx 0.618034 +\eps$
for any $\eps > 0$.
\end{theorem}
\begin{proof}
Our construction to establish this bound uses parameters $r>1$ and $a<1$ that we will fix later. 
We construct an instance such that any deterministic algorithm has competitive ratio at most $a$
for some prefix of the request sequence.
The instance has one {\em fast} machine of speed $r>1$ and $n$ {\em slow} machines of speed $1$.
The request sequence has non-decreasing job sizes:
$1=w_0 \leq w_1 \leq \ldots \leq w_{n-1} \leq w_n$.

The instance will satisfy the following properties:
\begin{align}
a \cdot r &\geq 1 \label{first_cond}\\
\forall k=1 \ldots n, \mbox{\ \ \ \ } a \cdot (r \cdot w_k + \sum_{j=0}^{i-1} w_j) &\geq r \cdot w_{k-1} + w_k \label{mid_cond}\\ 
a \cdot (r \cdot w_n + \sum_{j=0}^{n-1} w_j) &\geq r \cdot w_n \label{last_cond}
\end{align}

\begin{lemma}
Properties (\ref{first_cond}) - (\ref{last_cond}) imply that any deterministic algorithm has competitive ratio at most $a$.
\end{lemma}
\begin{proof}
Consider the action of a deterministic algorithm $A$ on the instance we construct.
If $w_0$ is placed on a slow machine, (\ref{first_cond}) implies that the competitive ratio is at most $a$.
For $k=1,\ldots n$, note that the optimum value for the prefix of the request sequence ending at $w_k$ is 
$r \cdot w_k +  \sum_{j=0}^{i-1} w_j$.
We consider two cases:

\noindent{\bf Case 1:} Suppose $A$ places some job on a slow machine.
Let $k$ be the first job placed on a slow machine. The the value of $A$ for the prefix ending at $w_i$
is $r \cdot w_{k-1} + w_k$.
Now (\ref{mid_cond}) implies that the competitive ratio is at most $a$.

\noindent{\bf Case 2:} All jobs are placed on the fast machine.
Then, once the entire sequence is processed upto $w_n$, (\ref{last_cond}) implies that the competitive ratio 
is at most $a$.
\end{proof}

In order to produce the instance, we focus on satisfying the following properties instead, which imply the previous properties (\ref{first_cond}) - (\ref{last_cond}).

\begin{align}
a \cdot r &\geq 1 \label{first_cond2}\\
a \cdot (r \cdot w_1 + w_0) &= r \cdot w_0 + w_1 \label{init_cond2}\\
\forall i=2 \ldots n, \quad a \cdot (r \cdot w_k - r\cdot w_{k-1} + &w_{k-1})  \notag\\
= r \cdot w_{k-1} + w_k  - (&r \cdot w_{k-2} + w_{k-1} )\label{mid_cond2}\\ 
r \cdot w_{n-1} + w_n &\geq r \cdot w_n \label{last_cond2}
\end{align}

Note that (\ref{init_cond2}) and (\ref{mid_cond2}) together imply (\ref{mid_cond}) with equality.
Further, (\ref{last_cond2}) and (\ref{mid_cond2}) for $k=n$ implies  (\ref{last_cond}).

Next, we rewrite  (\ref{mid_cond2}) as a recurrence relation for the sequence $\{ w_k \}$.
\begin{align}
(a \cdot r -1) w_k - (a+1)(r-1) w_{k-1} + r \cdot w_{k-2} = 0 \label{recurrence}
\end{align}
Note that the initial conditions are $w_0=1$ and from (\ref{init_cond2}), $w_1 = (r-a)/(a\cdot r-1)$.

Let $\delta > 0$ be a sufficiently small constant. 
\begin{lemma}\label{claim:roots}
\begin{align*}
\text{For\ } a = \frac{1+\sqrt{5+12\delta+4\delta^2}}{3+\sqrt{5}+2\delta} \text{\ and\ }
r = \frac{1+\sqrt{5+12\delta+4\delta^2}}{3-\sqrt{5}+2\delta},
\end{align*}
the roots of the characteristic equation of recurrence relation (\ref{recurrence}) are
$\displaystyle \frac{1+\sqrt{5}}{2} \pm i \sqrt{\delta}$
\end{lemma}
\begin{proof}
We will choose parameters $a,r$ such that the roots of the characteristic equation of the recurrence relation 
(\ref{recurrence}) are of the form $z \pm i \sqrt{\delta}$ for some small $\delta$.
Note that $i$ here is the complex square root of $-1$.
The reason for this choice of roots will become clear later.

First, we derive relationships between $a,r$ and $z,\delta$. Using the standard formula for the roots of a quadratic equation, we get
\begin{align}
\frac{(a+1)(r-1)}{2(r \cdot a-1)} &= z \label{root1}\\
\left(\frac{(a+1)(r-1)}{2(r \cdot a-1)}\right)^2 - \frac{r}{r\cdot a - 1} &= -\delta \label{root2}\\
(\ref{root2}) \Rightarrow  \frac{r}{r\cdot a - 1} &= z^2+\delta \label{root3}\\
(\ref{root1}) \Rightarrow \frac{r-a}{r\cdot a - 1} &= 2z-1  \label{root4}\\
(\ref{root3}) - (\ref{root4}) \Rightarrow  \frac{a}{r\cdot a - 1} &= (z-1)^2+\delta \label{root5}
\end{align}
An easy calculation shows that the system of equations
$$\frac{r}{r\cdot a - 1} = x, \ \ \ \ \ 
\frac{a}{r\cdot a - 1} = y$$
has the solution
$\displaystyle a = \frac{1+\sqrt{1+4xy}}{2x}$,
$\displaystyle r = \frac{1+\sqrt{1+4xy}}{2y}$.
Substituting $x=z^2+\delta$ and $y=(z-1)^2+\delta$, we get
\begin{align*}
a &= \frac{1+\sqrt{1+4(z^2+\delta)((z-1)^2+\delta)}}{2(z^2+\delta)}\\
r &= \frac{1+\sqrt{1+4(z^2+\delta)((z-1)^2+\delta)}}{2((z-1)^2+\delta)}
\end{align*}
Recall that $a$ is the upper bound on the competitive ratio that we establish and we would like to minimize $a$ to get the best bound possible. In fact, $a$ is minimized for $z=(1+\sqrt{5})/2$ and $\delta=0$ for which we get $a=(\sqrt{5}-1)/2$.
In fact, our construction will need $\delta > 0$, but we fix $z=(1+\sqrt{5})/2$ to minimize the upper bound $a$ we obtain from this argument. For this value of $z$ we get\\
$\displaystyle a = \frac{1+\sqrt{5+12\delta+4\delta^2}}{3+\sqrt{5}+2\delta} \mbox{\ \ \ }
r = \frac{1+\sqrt{5+12\delta+4\delta^2}}{3-\sqrt{5}+2\delta}$.
\end{proof}

Note that the upper bound $a$ on the competitive ratio is of the form $(\sqrt{5}-1)/2+\eps$ with $\eps$ a suitable function of  $\delta$ (and $\eps \rightarrow 0$ as $\delta \rightarrow 0$).

\begin{claim}
For $a$ and $r$ chosen as in Lemma~\ref{claim:roots}, 
the term $w_k$ of the solution to the recurrence relation is given by
\begin{align}
w_k = &\left(\frac{1}{2}- \frac{\sqrt{5}-1}{4\sqrt{\delta}}\cdot i \right) \left(\frac{1+\sqrt{5}}{2} + i \sqrt{\delta}\right)^k \notag \\ 
&\quad + \left(\frac{1}{2}+\frac{\sqrt{5}-1}{4\sqrt{\delta}}\cdot i \right)  \left(\frac{1+\sqrt{5}}{2} - i \sqrt{\delta}\right)^k 
\label{rec_soln}
\end{align}
\end{claim}
\begin{proof}
A general term $w_k$ of the sequence is given by the following expression:
\begin{align*}
(b+c \cdot i) \left(\frac{1+\sqrt{5}}{2} + i \sqrt{\delta}\right)^k + (d+e \cdot i) \left(\frac{1+\sqrt{5}}{2} - i \sqrt{\delta}\right)^k 
\end{align*}
where $b,c,d,e$ are real numbers to be determined.
Recall the initial conditions $w_0 = 1$ and  $w_1 = (r-a)/(r \cdot a -1) = 2z-1 = \sqrt{5}$. This gives the following equations:
\begin{align*}
(b+c \cdot i) \cdot 1 + (d + e \cdot i) \cdot 1 &= 1\\
\begin{split}
 (b+c \cdot i) \left(\frac{1+\sqrt{5}}{2} + i \sqrt{\delta}\right) \\
  + (d+e \cdot i) \left(\frac{1+\sqrt{5}}{2} - i \sqrt{\delta}\right) 
  &= \sqrt{5}
\end{split}
\end{align*}
Solving these, we get $b=d=1/2$, $c = -(\sqrt{5}-1)/4\sqrt{\delta}$ and $e = (\sqrt{5}-1)/4\sqrt{\delta}$.
Hence the term $w_k$ of the sequence is given by\\
$\displaystyle \left(\frac{1}{2}- \frac{\sqrt{5}-1}{4\sqrt{\delta}}\cdot i \right) \left(\frac{1+\sqrt{5}}{2} + i \sqrt{\delta}\right)^k$\\
$\displaystyle + \left(\frac{1}{2}+\frac{\sqrt{5}-1}{4\sqrt{\delta}}\cdot i \right)  \left(\frac{1+\sqrt{5}}{2} - i \sqrt{\delta}\right)^k$
\end{proof}

\begin{lemma}
For $a$ and $r$ chosen as in Lemma~\ref{claim:roots} and $\delta>0$ sufficiently small,
the solution $w_k$ of the recurrence relation satisfies the conditions 
(\ref{first_cond2})-(\ref{last_cond2}).
\end{lemma}
\begin{proof}
The choice of recurrence relation (\ref{recurrence}) ensures that (\ref{mid_cond2}) is satisfied,
and the initial condition for $w_1$ ensures that (\ref{init_cond2}) is satisfied.
For $\delta=0$, $a \cdot r = (6+2\sqrt{5})/4$.
Hence for $\delta$ sufficiently small, $a \cdot r \geq 1$ and hence (\ref{first_cond2}) is satisfied.

Rewriting condition (\ref{last_cond2}) , we need to show that $w_n / w_{n-1} \leq r/(r-1)$.
Note that $r=\frac{1+\sqrt{5}}{3-\sqrt{5}} + O(\delta)$ and $r/(r-1) \ge (3+\sqrt{5})/4 - O(\delta) \ge 1.309$.
Examining the solution (\ref{rec_soln}) of the recurrence, we see that $w_k$ is twice the real part of the
first term in (\ref{rec_soln}).
In other words, $w_k$ is of the form
\begin{align}
w_k &= 2 \Re \left( (r_1 e^{i \Phi_1})(r_2 e^{i \Phi_2})^k \right)
\end{align}
where 
$\displaystyle \Phi_1 = \tan^{-1}\left( -\frac{\sqrt{5}-1}{2\sqrt{\delta}} \right)$,
$\displaystyle \Phi_2 = \tan^{-1}\left( \frac{2\sqrt{\delta}}{1+\sqrt{5}}\right)$, and $r_2 = \frac{1+\sqrt{5}}{2}+O(\delta)$. 
As $\delta \rightarrow 0$, 
$$\Phi_1 = -\frac{\pi}{2} + \frac{2\sqrt{\delta}}{\sqrt{5}-1} - O(\delta^{3/2})$$
and
$\displaystyle \Phi_2 = \frac{2\sqrt{\delta}}{1+\sqrt{5}} - E$ for $E=O(\delta^{3/2})$.
Note that
\begin{align}
w_k &= 2 r_1 (r_2)^k cos (\Phi_1 + k \Phi_2)\\
 &= 2 r_1 \left(\frac{1+\sqrt{5}}{2}+O(\delta)\right)^k \notag \\
&\cos \left(-\frac{\pi}{2} +\frac{2\sqrt{\delta}}{\sqrt{5}-1} + k \frac{2\sqrt{\delta}}{1+\sqrt{5}} -O(\delta^{3/2}) - kE\right)
\end{align}
\begin{align}
&\frac{w_n}{w_{n-1}} = \left(\frac{1+\sqrt{5}}{2}+O(\delta)\right)\notag\\
&\frac{\cos \left( -\frac{\pi}{2} +\frac{2\sqrt{\delta}}{\sqrt{5}-1} + n \frac{2\sqrt{\delta}}{1+\sqrt{5}} -O(\delta^{3/2}) - nE\right)}
{\cos \left( -\frac{\pi}{2} +\frac{2\sqrt{\delta}}{\sqrt{5}-1} + (n-1) \frac{2\sqrt{\delta}}{1+\sqrt{5}} -O(\delta^{3/2}) - (n-1)E\right) }
\end{align}
Note that we hope to show $w_n / w_{n-1} \le 1.309 \le (3+\sqrt{5})/4-O(\delta)$.
The term $(1+\sqrt{5})/2+O(\delta) < 1.619$, so it is critical that the ratio of cosines be small enough
to give us the condition we want (but not too small to ensure that $w_n/w_{n-1} \geq 1$).

First we show that the smallest ratio of consecutive terms is $w_n/w_{n-1}$. Consider\\
$f(x) = \frac{\cos(x+\Phi_2)}{\cos(x)} = \cos(\Phi_2) - \tan(x)\sin(\Phi_2)$.\\ 
We have $f'(x) = -\frac{\sin(\Phi_2)}{\cos^2(x)} < 0$ so the ratio gets smaller as $n$ increases, as long as $\Phi_1 + (n-1)\Phi_2 < \pi/2$.

Next we show there exists $n$ such that \\
$13/20 \le \frac{\cos(\Phi_1+n\Phi_2)}{\cos(\Phi_1+(n-1)\Phi_2)} \le 3/4$. For $n = 1$, the ratio is greater than 1 for sufficiently small $\delta$. For the smallest $n$ such that $\Phi_1 + n\Phi_2 \ge \pi/2$, the ratio is smaller than 0. Furthermore, if $\frac{\cos(\Phi_1+n\Phi_2)}{\cos(\Phi_1+(n-1)\Phi_2)} > 3/4$ then 
\begin{align*}
&\frac{\cos(\Phi_1+(n+1)\Phi_2)}{\cos(\Phi_1+n\Phi_2)} \\
&\qquad = \frac{2\cos(\Phi_1+n\Phi_2)\cos(\Phi_2) - \cos(\Phi_1+(n-1)\Phi_2)}{\cos(\Phi_1+n\Phi_2)}\\ 
&\qquad \ge 2\cos(\Phi_2) - \frac{4}{3} \ge 13/20 \qquad \text{if $\cos(\Phi_2) \ge \frac{119}{120}$}
\end{align*}
This is true for sufficiently small $\delta$.
\end{proof}
\end{proof}
\section{Upper bound for randomized algorithms}\label{sec:randlb}
To establish the bound for randomized algorithms, we use Yao's principle and show an upper bound on the expected competitive ratio of any deterministic algorithm on a distribution of instances. The construction uses one fast machine of speed 1 and $n$ slow machines of speed 1/4. The request sequence has non-decreasing sizes $2^i$. The prefix of this sequence ending with size $2^i$ is presented to the algorithm with probability $c/2^i$, where $c$ is a normalizing constant. We show that the best algorithm for this sequence achieves at most $cn+1$ while the optimal algorithm achieves roughly $5nc/4$. 
\begin{theorem}\label{thm:randlb}
The competitive ratio of any randomized algorithm against an oblivious adversary is at most $0.8+\eps$
for any $\eps > 0$.
\end{theorem}
\begin{proof}
In order to establish the bound for randomized algorithms,
we use Yao's principle and show an upper bound on the expected competitive ratio of 
any deterministic algorithm on a distribution of instances.
The construction uses a set of machines with one {\em fast} machine of speed 1 and
$n$ {\em slow} machines of speed $1/4$.
The request sequence has non-decreasing sizes $w_i = 2^i$ for $i = 1, \ldots, n$.
Our construction uses a probability distribution over prefixes of this sequence:
the prefix ending at $w_i$ is presented to the algorithm with probability $p_i = c/2^i$,
where the normalizing constant $c =1/(1-1/2^n)$.

Let $OPT_i$ denote the optimal solution for the length $i$ prefix of the input.
It is easy to see that $OPT_i$ places $w_i$ on the fastest machine and jobs
$w_1, \ldots, w_{i-1}$ on the slow machines.
Hence $OPT_i = 2^i + (2^i-1)/4 = (5/4)2^i - 1/4$.
We will compute the expected value of the optimal solution for the distribution on
inputs specified above.
\begin{align*}
\sum_{i=1}^n \frac{c}{2^i} OPT_i &= c \sum_{i=1}^n \frac{(5/4)2^i - 1/4}{2^i} 
= c \sum_{i=1}^n \frac{5}{4} - \frac{1}{4\cdot 2^i} \\
&= \left(1-\frac{1}{2^n}\right)^{-1}\left(\frac{5n}{4} - \frac{1}{4}(1-1/2^n)\right)
\end{align*}

Next we compute the expected value of the best deterministic algorithm on this distribution. Notice that for the setting we specified, a deterministic algorithm is completely specified by $n$ choices of whether to put $w_i$ on the fast machine or to put it on an unoccupied slow machine. Let $c_i$ be the indicator variable of whether the algorithm puts $w_i$ on the fast machine. Let $a_1 \le \cdots \le a_k$ be the indices of the jobs the algorithm puts on the fast machine. Let $m_i$ be the maximum $a_l$ such that $a_l \le i$. In other words, $m_i$ is the index of the largest job on the fastest machine if the sequence of jobs ends at the $i$th job. For notational convenience, assume that we already put a job of size $w_0 = 0$ on the fast machine and $a_0 = 0$. The expected value of the algorithm is
\begin{align*}
&\sum_{i=1}^n p_i \left(w_{m_i} + \sum_{j:j \not\in \{a_1, \cdots, a_k\}\wedge j \le i} \frac{w_j}{4}\right)\\ 
&\qquad = \sum_{i=1}^n p_i \left(\sum_{j:a_j\le i}(w_{a_j}-w_{a_{j-1}}) + \sum_{j\textnormal{: }\bar{c_j}\wedge j \le i} \frac{w_j}{4}\right)\\
&\qquad =\sum_{j=1}^n \left(c_j(w_j - w_{m_{j-1}}) + (1-c_j)\frac{w_j}{4}\right)\sum_{i\ge j}p_i\\
&\qquad =\sum_{j=1}^n x_j
\end{align*}
where $x_j = \left(c_j(w_j - w_{m_{j-1}}) + (1-c_j)\frac{w_j}{4}\right)\sum_{i\ge j}p_i$.
Let $f(t,j)$ be the maximum value of $\sum_{l=1}^t x_l$ over all choices of $c_1, \cdots, c_t$ with the restriction that $m_t = j$. First, if $j < t$ then $c_{j+1}=\cdots =c_{t} = 0$ and we have 
\begin{equation}f(t,j) = f(t-1,j) + \frac{w_t}{4}\sum_{i\ge t}p_i \le f(t-1,j) + \frac{c}{2} \label{eqn:t-ge-j}\end{equation}
Next, if $j = t$ then
\begin{align}
\nonumber f(t,t) &= \max_{m_{t-1}} \left( f(t-1, m_{t-1}) + (w_t - w_{m_{t-1}})\sum_{i\ge t}p_i \right)\\
&\le\max_{m_{t-1}} \left( f(m_{t-1}, m_{t-1}) + (w_t - w_{m_{t-1}})\frac{2c}{2^t} \right. \notag\\
&\left. \qquad\qquad\qquad\qquad\qquad\quad + \frac{(t-1-m_{t-1})c}{2} \right) \label{eqn:t-eq-j}
\end{align}
The above inequality follows from $\sum_{i\ge j}p_i \le \frac{2c}{2^j}$ and Equation~\ref{eqn:t-ge-j}.

We now prove by induction that $f(i,i) \le ci+1$. The base cases $f(0,0) = 0$, $f(1,1) = 2$ and $f(2,2) = 4-c$ are obvious. Assume that the claim holds up to $i=t-1$ and we want to prove it for $i=t$. Consider 3 cases for $m_{t-1}$. First, if $m_{t-1} = t-1$ then $f(t,t) = f(t-1,t-1) + (2^t-2^{t-1})\frac{2c}{2^t} \le c(t-1)+1+c = ct+1$. Next, if $m_{t-1}=t-2$ then $f(t,t) = f(t-2,t-2) + (2^t - 2^{t-2})\frac{2c}{2^t} + c/2 \le ct+1$. Lastly, if $m_{t-1} \le t-3$ then
\begin{align*}
f(t,t) &\le f(m_{t-1}, m_{t-1}) + w_t \frac{2c}{2^t} + (t-1-m_{t-1})c/2\\
&\le m_{t-1}c + 1 + 2c + (t-1-m_{t-1})c/2 \le tc+1
\end{align*}
Thus, we have proved the inductive case. By Equation~\ref{eqn:t-ge-j} and the fact that $f(i,i) \le ci+1$, we have $f(n, j) \le cn+1~\forall j$. However, $\max_j f(n,j)$ is exactly the expected value of the best algorithm so the competitive ratio of any randomized algorithm on the specified instance is at most $\frac{cn+1}{5nc/4 - 1/4}\rightarrow 4/5$ as $n$ goes to infinity.

\end{proof}

\subsubsection{Acknowledgments.} MC was supported by NSF awards CCF 0832797, AF 0916218 and a Google research award. MH's support: The research leading to these results has received funding from the European Research Council 
under the European Union's Seventh Framework Programme (FP/2007-2013) / ERC Grant
Agreement no. 340506 and  the Austrian Science Fund (FWF) grant P23499-N23. HN was supported by NSF awards CCF 0832797, AF 0916218, a Google research award, and a Gordon Wu fellowship.
\bibliographystyle{abbrv}
\bibliography{../matching}
\end{document}